\documentclass[journal]{IEEEtran}
\usepackage[utf8]{inputenc}
\usepackage{amsmath}
\usepackage{amsfonts,color}
\usepackage{amsthm}
\usepackage{graphicx,subcaption}
\usepackage{bbm}

\newtheorem{thm}{Theorem}

\newtheorem{cor}{Corollary}
\newtheorem{lem}{Lemma}

\newcommand{\calP}{{\mathcal{P}}}

\newcommand{\calN}{{\mathcal{N}}}
\def \bE {\mathbb{E}}
\def \bR {\mathbb{R}}

\title{Lower Bounds for Learning Distributions under Communication Constraints via Fisher Information}
\author{Leighton Pate Barnes, Yanjun Han, and Ayfer \"Ozg\"ur\\
Stanford University, Stanford, CA 94305\\
 Email: \{lpb, yjhan, aozgur\}@stanford.edu}

\begin{document}

\maketitle

\begin{abstract}
	
We consider the problem of learning high-dimensional, nonparametric and structured (e.g. Gaussian) distributions in distributed networks, where each node in the network observes an independent sample from the underlying distribution and can use $k$ bits to communicate its sample to a central processor. We consider three different models for communication. Under the independent model, each node communicates its sample to a central processor by independently encoding it into $k$ bits. Under the more general sequential or blackboard communication models, nodes can  share information interactively but each node is restricted to write at most $k$ bits on the final transcript. We characterize the impact of the communication constraint $k$ on the minimax risk of estimating the underlying distribution under $\ell^2$ loss. We develop minimax lower bounds that apply in a unified way to many common statistical models and reveal that the impact of the communication constraint can be qualitatively different depending on the tail behavior of the score function associated with each model. A key ingredient in our proofs is a geometric characterization of Fisher information from quantized samples.
\end{abstract}

\section{Introduction}
Estimating a distribution from samples is a fundamental unsupervised learning
problem that has been studied in statistics since the late nineteenth century \cite{hist}. Consider the following distribution estimation model
	\begin{align*}
	X_1, X_2, \cdots, X_n \overset{\text{i.i.d.}}{\sim} P, 
	\end{align*}
	where we would like to estimate the unknown distribution $P$ under squared $\ell^2$ loss. Unlike the traditional statistical setting where samples $X_1,\cdots,X_n$ are available to the estimator as they are, in this paper we consider a distributed setting where each observation $X_i$ is available at a different node and has to be communicated to a central processor by using $k$ bits. We consider three different types of communication protocols: 
	\begin{enumerate}
		\item Independent communication protocols $\Pi_{\mathsf{Ind}}$: each node sends a $k$-bit string $M_i$ simultaneously (independent of the other nodes) to the central processor and the final transcript is $Y=(M_1,\ldots,M_n)$;
		\item Sequential communication protocols $\Pi_{\mathsf{Seq}}$: the nodes sequentially send $k$-bit strings $M_i$, where quantization of the sample $X_i$ can depend on the previous messages $M_1,\ldots,M_{i-1}$;
		\item Blackboard communication protocols $\Pi_{\mathsf{BB}}$ \cite{kush}: all nodes communicate via a publicly shown blackboard while the total number of bits each node can write in the final transcript $Y$ is limited by $k$. When one node writes a message (bit) on the blackboard, all other nodes can see the content of the message and depending on the written bit another node can take the turn to write a message on the blackboard. 
	\end{enumerate}
	Upon receiving the transcript $Y$, the central processor produces an estimate $\hat{P}$ of the distribution $P$ based on the transcript $Y$ and known procotol $\Pi$ which can be of type $\Pi_{\mathsf{Ind}}$, $\Pi_{\mathsf{Seq}}$, or $\Pi_{\mathsf{BB}}$. Our goal is to jointly design the  protocol $\Pi$ and the estimator $\hat{P}(Y)$ so as to minimize the worst case squared $\ell^2$ risk, i.e., to characterize
$$
\inf_{(\Pi,\hat{P})}\sup_{P\in \calP} \bE_{P}\|\hat{P}-P\|_2^2,
$$
where $\calP$ denotes the class of distributions which $P$  belongs to. We study three different instances of this estimation problem:
\begin{enumerate}
\item High-dimensional discrete distributions: in this case we assume that $P=(p_1,\cdots,p_d)$ is a discrete distribution with known support size $d$ and $\calP$ denotes the probability simplex over $d$ elements. By ``high-dimensional'' we mean that the support size $d$ of the underlying distribution may be comparable to the sample size $n$.
\item Non-parametric densities: in this case $X_1, \cdots, X_n \overset{\text{i.i.d.}}{\sim} f$, where $f$ is some density that possesses some standard H\"older continuity property \cite{nemirovski2000topics}. 
\item Parametric distributions: in this case, we assume that we have some additional information regarding the structure of the underlying distribution or density. In particular, we assume that the underlying distribution or density can be parametrized such that
\begin{align*}
X_1, X_2, \cdots, X_n \overset{\text{i.i.d.}}{\sim} P_\theta, 
\end{align*}
where $\theta\in\Theta\subset \bR^d$. In this case, estimating the underlying distribution can be thought of as estimating the parameters of this distribution, and we are interested in the following parameter estimation problem under squared $\ell^2$ risk
$$
\inf_{(\Pi, \hat{\theta})}\sup_{\theta\in\Theta} \bE_{\theta}\|\hat{\theta}-\theta\|_2^2,
$$
where $\hat{\theta}(\cdot)$ is an estimator of $\theta$.
\end{enumerate}	

Statistical estimation in distributed settings has gained increasing popularity over the recent years motivated by the fact that modern data sets are often distributed across multiple machines and processors, and bandwidth and energy limitations in networks and within multiprocessor systems often impose significant bottlenecks on the performance of algorithms. There are also an increasing number of applications in which data is generated in a distributed manner and the data (or features of it) are communicated over bandwidth-limited links to central processors. For example, there is a recent line of works \cite{duchi,braverman, garg} which focus on the distributed parameter estimation problem where the underlying distribution has a Gaussian structure, i.e. $P_\theta=\calN(\theta,I_d)$ with $\theta\in\Theta\subseteq\bR^d$, often called the Gaussian location model. The high-dimensional discrete distribution estimation problem is studied in \cite{diakonikolas,yanjun2}, where extensions to distributed property testing are studied in \cite{archayaetal,archayaetal2}. These works show that the dependence of the estimation performance on $k$ can be qualitatively different: the estimation error decays linearly in $k$ for the Gaussian location model, while for distribution estimation/testing it typically decays exponentially in $k$. This difference was first studied in \cite{yanjun} which develops geometric lower bounds for distributed estimation, where the Gaussian mean estimation problem and distribution estimation problem admit different geometric structures. However, the arguments heavily rely on hypothesis testing and the specific geometric objects remain implicit. 

Another closely-related thread is the locally private estimation problem, which shares many similarities with communication-constrained problems \cite{duchi2019lower}. We refer to \cite{duchi2013local} for a general treatment of estimation problems under locally differentially private (LDP) constraints, while optimal schemes (and lower bounds) for estimating discrete distributions are proposed in \cite{kairouz2016discrete,wang2016mutual,ye2018optimal,acharya2018hadamard}. Similar to the previous discussions, strong or distributed data-processing inequalities \cite{duchi2013local,raginski} are typically used in scenarios where the linear/quadratic dependence on the pricacy parameter $\varepsilon$ is tight, and explicit modeling becomes necessary in scenarios where the tight dependence on $\varepsilon$ is exponential. 

In this paper, we approach all distributed estimation problems under communication constraints in a unified way. Specifically, we propose an explicit geometric object related to the Fisher information, and develop a framework that characterizes the Fisher information for estimating an underlying unknown parameter from a quantized sample. Equivalently, we ask the question:  how can we best represent $X\sim P_\theta$ with $k$ bits so as to maximize the Fisher information it provides about the underlying parameter $\theta$?  This framework was first introduced in \cite{allerton,isit}, and there has been some previous work in analyzing Fisher information from a quantized scalar random variable such as \cite{fisher_quant1,fisher_quant2,fisher_quant3,fisher_quant4}. Different from these works, here we consider the arbitrary quantization of a random vector and are able to study the impact of the quantization rate $k$ along with the dimension $d$ of the underlying statistical model on the Fisher information. As an application of our framework, we use upper bounds on Fisher information to derive lower bounds on the minimax risk of the distributed estimation problems discussed above and recover many of the existing results in the literature \cite{duchi,braverman, garg, yanjun2}, which are known to be rate-optimal. Our technique is significantly simpler and more transparent than those in \cite{duchi,braverman, garg, yanjun2}. In particular, the strong/distributed data processing inequalities used in \cite{duchi,braverman, garg} are typically technical and seem to be only applicable to models where the fundamental dependence of the minimax risk on the quantization rate $k$ is linear, e.g., the Gaussian location model. Moreover, our approach points out that the Fisher information is the same as the explicit geometric object from \cite{yanjun}, and we recover most results in \cite{yanjun} via this simpler approach. We also extend the results of \cite{yanjun} to derive minimax lower bounds for statistical models with sub-exponential score functions, which is useful, for example, when we are interested in estimating the variance of a Gaussian distribution. 

\subsection{Organization of the Paper}
In the next section, we introduce the problem of characterizing Fisher information from a quantized sample. We present a geometric characterization for this problem and derive two upper bounds on Fisher information as a function of the quantization rate. We also evaluate these upper bounds for common statistical models. In Section~\ref{sec:est}, we formulate the problem of distributed learning of distributions under communication constraints with  independent, sequential and blackboard communication protocols. We use the upper bounds on Fisher information from Section~\ref{fisher} to derive lower bounds on the minimax risk of distributed estimation of discrete and parametric distributions. There we also provide a more detailed comparison of our results with those in the literature. Finally, in Section~\ref{sec:nonpar} we discuss extending these results to non-parametric density estimation.

\section{Fisher information from a quantized sample}\label{fisher}
Let $\{P_\theta\}_{\theta\in\Theta}$ be a family of probability measures on the measurable space $(\mathcal{X},\mathcal{A})$ parameterized by $\theta\in\Theta\subseteq\mathbb{R}^d$. Suppose that $\{P_\theta\}_{\theta\in\Theta}$ is a dominated family and that each $P_\theta$ has density $f(x|\theta)$ with respect to some dominating measure $\nu$. Let $X\in\mathcal{X}$ be a single sample drawn from $f(x|\theta)$. Any (potentially randomized) $k$-bit quantization strategy for $X$ can be expressed in terms of the conditional probabilities
$$
b_m(x)=p(m|x)\quad \text{for}\quad m\in[2^k],\quad x\in\mathcal{X}. 
$$ 
We assume 
that $p(m|x)$ is a regular conditional probability. Under any given $P_\theta$ and quantization strategy, denote by $p(m|\theta)$ the likelihood that the quantized sample $M$ takes a specific value $m$. Let
\begin{align*}
S_\theta(m) & = \left(S_{\theta_1}(m),\ldots,S_{\theta_d}(m) \right) \\
& = \left(\frac{\partial}{\partial\theta_1}\log p(m|\theta),\ldots,\frac{\partial}{\partial\theta_d}\log p(m|\theta) \right)  \in \bR^d
\end{align*}
be the vector-valued score function of $M$ under $P_\theta$. With a slight abuse of notation, we also denote the score function of $X$ under $P_\theta$ as
\begin{align*}
S_\theta(x) & =  \left(S_{\theta_1}(x),\ldots,S_{\theta_d}(x) \right) \\
& =  \left(\frac{\partial}{\partial\theta_1}\log f(x|\theta),\ldots,\frac{\partial}{\partial\theta_d}\log f(x|\theta) \right)\in \bR^d \; .
\end{align*}
Consequently, the Fisher information matrices of estimating $\theta$ from $M$ and from $X$ are defined as
$$I_M(\theta) = \mathbb{E}[S_\theta(M) S_\theta(M)^T], $$
$$I_X(\theta) = \mathbb{E}[S_\theta(X) S_\theta(X)^T],$$
respectively. 

We will assume throughout that $f(x|\theta)$ satisfies the following regularity conditions:
\begin{itemize}
\item[(1)] The function $\theta\mapsto \sqrt{f(x|\theta)}$ is continuously differentiable coordinate-wise at $\nu$-almost every $x\in\mathcal{X}$; 
\item[(2)] For all $\theta$, the Fisher information matrix $I_X(\theta)$ exists and is continuous  coordinate-wise in $\theta$.
\end{itemize}
These two conditions justify interchanging differentiation and integration as in
\begin{align*}
\nabla_\theta p(m|\theta) & = \nabla_\theta \int f(x|\theta)p(m|x)d\nu(x) \\
& = \int \nabla_\theta f(x|\theta)p(m|x)d\nu(x), 
\end{align*}
and they also ensure that $p(m|\theta)$ is continuously differentiable w.r.t. $\theta$ coordinate-wise (cf. \cite[Section 26, Lemma 1]{borovkov}). 

The following two lemmas establish a geometric interpretation of the trace $\mathsf{Tr}(I_M(\theta))$, and are slight variants of \cite[Theorems 1 and 2]{amari}. 
\begin{lem} \label{lem:lem1}
For $i\in [d]$, the $(i,i)$-th entry of the Fisher information matrix $I_M(\theta)$ is
$$[I_M(\theta)]_{i,i} = \mathbb{E}\left[\mathbb{E}\left[S_{\theta_i}(X)|M\right]^2\right], $$
where the inner conditional expectation is with respect to the distribution $f(x|\theta,m)$, and the outer expectation is over the conditioning random variable $M$.
\end{lem}
\begin{proof}
	For any $m\in [2^k]$, we have
\begin{align*}
\mathbb{E}_\theta\left[S_{\theta_i}(X)|m\right] & = \int S_{\theta_i}(x) \frac{f(x|\theta)p(m|x)}{p(m|\theta)}d\nu(x)\\
& =  \int \frac{\frac{\partial}{\partial\theta_i}f(x|\theta)}{f(x|\theta)}\frac{f(x|\theta)p(m|x)}{p(m|\theta)}d\nu(x) \\
& =  \frac{1}{p(m|\theta)}\int \frac{\partial}{\partial\theta_i}f(x|\theta)p(m|x)d\nu(x) \\
& =  \frac{1}{p(m|\theta)}\frac{\partial}{\partial\theta_i} \int f(x|\theta)p(m|x)d\nu(x) \\
& =  S_{\theta_i}(m). 
\end{align*}
Squaring both sides and taking expectation with respect to $M$ completes the proof. 
\end{proof}

\begin{lem} \label{lem:lem2}
The trace of the Fisher information matrix $I_M(\theta)$ can be written as
\begin{align}
\mathsf{Tr}(I_M(\theta)) = \sum_{m\in [2^k]} p(m|\theta)\|\mathbb{E}[S_\theta(X)|m]\|_2^2  \label{eq:lem2} \; .
\end{align}
\end{lem}
\begin{proof}
By Lemma 1,
\begin{align*} 
\sum_{i=1}^d [I_M(\theta)]_{i,i}  & = \sum_{i=1}^d \nonumber \mathbb{E}\left[\mathbb{E}\left[S_{\theta_i}(X)|M\right]^2\right] \nonumber \\
& = \mathbb{E}\left[ \sum_{i=1}^d\mathbb{E}\left[S_{\theta_i}(X)|M\right]^2\right] \nonumber \\
& = \mathbb{E}\left[\|\mathbb{E}[S_\theta(X)|M]\|_2^2\right] \nonumber \\
& = \sum_m p(m|\theta)\|\mathbb{E}[S_\theta(X)|m]\|_2^2\; .
\end{align*}
\end{proof}

In order to get some geometric intuition for the quantity \eqref{eq:lem2}, consider a special case where the quantization is deterministic and the score function $S_\theta(x)$ is a bijection between $\mathcal{X}$ and $\bR^d$. In this case, the quantization map partitions the space $\mathcal{X}$ into disjoint quantization bins, and this induces a corresponding partitioning of the score functions values $S_\theta(x)$. Each vector $\mathbb{E}[S_\theta(X)|m]$ is then the centroid of the set of $S_\theta(x)$ values corresponding to quantization bin $m$  (with respect to the induced probability distribution on $S_\theta(X)$). Lemma \ref{lem:lem2} shows that $\mathsf{Tr}(I_M(\theta))$ is equal to the average squared magnitude of these centroid vectors. 


\subsection{Upper Bounds on $\mathsf{Tr}(I_M(\theta))$}
 In this section, we give two upper bounds on $\mathsf{Tr}(I_M(\theta))$ depending on the different tail behaviors of $S_\theta(X)$, with proofs deferred to Appendix \ref{app:proofs}. The first theorem upper bounds $\mathsf{Tr}(I_M(\theta))$ in terms of the variance of $S_\theta(X)$ when projected onto any unit vector. 

\begin{thm} \label{thm:main_thm}
If for any $\theta\in\Theta$ and any unit vector $u\in\mathbb{R}^d$, $$\mathsf{Var}(\langle u, S_\theta(X)\rangle) \leq I_0 \; ,$$  then
$$\mathsf{Tr}(I_M(\theta)) \leq \min\{ \mathsf{Tr}(I_X(\theta)), \; 2^kI_0\} \; .$$
\end{thm}

The first upper bound $\mathsf{Tr}(I_M(\theta)) \leq \mathsf{Tr}(I_X(\theta))$ follows easily from the data processing inequality for Fisher information \cite{zamir}. The second upper bound in Theorem \ref{thm:main_thm} shows that when $I_0$ is finite, the trace $\mathsf{Tr}(I_M(\theta))$ can increase at most exponentially in $k$. 

Our second theorem upper bounds $\mathsf{Tr}(I_M(\theta))$ in terms of the $\Psi_p$ Orlicz norm of $S_\theta(X)$ when projected onto any unit vector.  Recall that for $p \geq 1$, the $\Psi_p$ Orlicz norm of a random variable $X$ is defined as
$$\|X\|_{\Psi_p} = \inf\{K \in (0,\infty) \; | \; \mathbb{E}[\Psi_p(|X|/K)]\leq1\},$$
where
$$\Psi_p(x) = \exp(x^p) - 1 \; .$$
Note that a random variable with finite $\Psi_1$ Orlicz norm is sub-exponential, while a random variable with finite $\Psi_2$ Orlicz norm is sub-Gaussian \cite{versh}. 

\begin{thm} \label{thm:main_thm2}
If for any $\theta\in\Theta$ and any unit vector $u\in\mathbb{R}^d$, $$\|\langle u, S_\theta(X)\rangle\|_{\Psi_p}^2 \leq I_0\; $$ 
holds for some $p\ge 1$, then
$$\mathsf{Tr}(I_M(\theta)) \leq \min\{ \mathsf{Tr}(I_X(\theta)), \; Ck^{\frac{2}{p}} I_0\}, $$
where $C =4 \; .$
\end{thm}

Theorem \ref{thm:main_thm2} shows that when the score function $S_\theta(X)$ has a lighter tail, then the trace $\mathsf{Tr}(I_M(\theta))$ can increase at most polynomially in $k$ at the rate $O(k^\frac{2}{p})$. 

%
%
\subsection{Applications to Common Statistical Models}\label{sec:apps}
We next apply the above two results to common statistical models. We will see that depending on the statistical model, either bound may be tighter. The proofs of Corollaries \ref{cor:ex1} through \ref{cor:ex4} appear in Appendix \ref{app:proofs2}. In the next section, we show that Corollaries  \ref{cor:ex1}, \ref{cor:ex3}, \ref{cor:ex4} yield tight results for the minimax risk of the corresponding distributed estimation problems.

For the Gaussian location model, Corollary \ref{cor:ex1} follows by showing that the score function associated with this model has finite $\Psi_2$ Orlicz norm and applying Theorem~\ref{thm:main_thm2}.
\begin{cor}[Gaussian location model] \label{cor:ex1}
Consider the Gaussian location model $X\sim\mathcal{N}(\theta,\sigma^2I_d)$ where we are trying to estimate the mean $\theta$ of a $d$-dimensional Gaussian random vector with fixed covariance $\sigma^2I_d$. In this case,
\begin{equation} \label{eq:gaussian_location}
\mathsf{Tr}(I_M(\theta)) \leq \min\left\{\frac{d}{\sigma^2},C\frac{k}{\sigma^2}\right\}
\end{equation}
where
$$C = \frac{32}{3} \; .$$
\end{cor}

For covariance estimation in the independent Gaussian sequence model, Corollary \ref{cor:ex2} follows by showing that the score function associated with this model has finite $\Psi_1$ Orlicz norm and applying Theorem~\ref{thm:main_thm2}.
\begin{cor}[Gaussian covariance estimation]\label{cor:ex2}
Suppose $X\sim\mathcal{N}(0,\text{\rm diag}(\theta_1,\ldots,\theta_d))$ and $\Theta \subseteq [\sigma_{\min}^2,\sigma_{\max}^2]^d$ with $\sigma_{\max} > \sigma_{\min}>0$. In this case,
$$\mathsf{Tr}(I_M(\theta)) \leq \min\left\{\frac{d}{2\sigma_{\min}^4}, C\left(\frac{k}{\sigma_{\min}^2}\right)^2\right\}$$
where $$C= \frac{16(\log 4 + 2(2+\sqrt{2}))^2}{(\log 2)^2} \; .$$
\end{cor}

For distribution estimation, Corollary \ref{cor:ex3} is a consequence of Theorem~\ref{thm:main_thm} along with characterizing the variance to the score function associated with this model. 
\begin{cor}[Distribution estimation]\label{cor:ex3}
Suppose that $\mathcal{X} = \{1,\ldots,d+1\}$ and that 
$$f(x|\theta) = \theta_x \; .$$
Let $\theta_1,\ldots,\theta_d$ be the free parameters of interest and suppose they can vary from $\frac{1}{4d}\leq\theta_i\leq\frac{1}{2d}$. In this case,
\begin{equation}\label{eq:dist_est}
\mathsf{Tr}(I_M(\theta)) \leq  6\min\{d^2,d2^k\} \; .
\end{equation}
\end{cor}

For the product Bernoulli model, the tightness of Theorems \ref{thm:main_thm} and \ref{thm:main_thm2} differ in different parameter regions, as shown in the following Corollary \ref{cor:ex4}. 
\begin{cor}[Product Bernoulli model] \label{cor:ex4}
Suppose that $X\sim \prod_{i=1}^d\text{\rm Bern}(\theta_i)$. If $\Theta = [1/2-\varepsilon,1/2+\varepsilon]^d$ for some $0<\varepsilon<1/2$, i.e. the model is relatively dense, then
$$\mathsf{Tr}(I_M(\theta)) \leq C\min\{d,k\}$$
for some constant $C$ that depends only on $\varepsilon$. If $\Theta = [(\frac{1}{2}-\varepsilon)\frac{1}{d},(\frac{1}{2}+\varepsilon)\frac{1}{d}]^d$, i.e. the model is relatively sparse, then
$$\mathsf{Tr}(I_M(\theta)) \leq \frac{2d}{\frac{1}{2}-\varepsilon}\min\{d,2^k\} \; .$$
\end{cor}

In the product Bernoulli model,
$$S_{\theta_i}(x) = \begin{cases} \frac{1}{\theta_i}, & \; x_i = 1 \\ -\frac{1}{1-\theta_i}, & \; x_i=0 \end{cases}.$$ 
Hence, when $\Theta = [1/2-\varepsilon,1/2+\varepsilon]^d$,  $\mathsf{Var}(\langle u,S_\theta(X)\rangle)$ and $\|\langle u,S_\theta(X)\rangle\|_{\Psi_2}^2$ are both $\Theta(1)$. In this case, Theorem \ref{thm:main_thm} gives
$$\mathsf{Tr}(I_M(\theta)) = O(2^k), $$
while Theorem \ref{thm:main_thm2} gives
$$\mathsf{Tr}(I_M(\theta)) = O(k) \; .$$
In this situation Theorem \ref{thm:main_thm2} gives the better bound. On the other hand, if $\Theta = [(\frac{1}{2}-\varepsilon)\frac{1}{d},(\frac{1}{2}+\varepsilon)\frac{1}{d}]^d$, then $\mathsf{Var}(\langle u,S_\theta(X)\rangle) = \Theta(d)$ and $\|\langle u,S_\theta(X)\rangle\|_{\Psi_2}^2 = \Theta(d^2)$. In this case Theorem \ref{thm:main_thm} gives
$$\mathsf{Tr}(I_M(\theta)) = O(d2^k), $$
while Theorem \ref{thm:main_thm2} gives
$$\mathsf{Tr}(I_M(\theta)) = O(d^2k) \; .$$
In the sparse case $\mathsf{Tr}(I_M(\theta))\le \mathsf{Tr}(I_X(\theta)) = \Theta(d^2)$, so only the bound from Theorem \ref{thm:main_thm} is non-trivial. It is interesting that Theorem \ref{thm:main_thm2} is able to use the sub-Gaussian structure in the first case to yield a better bound -- but in the second case, when the tail of the score function is essentially not sub-Gaussian, Theorem \ref{thm:main_thm} yields the better bound.

\section{Distributed Parameter Estimation} \label{sec:est}
In this section, we apply the results in the previous section to the distributed estimation of parameters of an underlying statistical model under communication constraints. In this section we only focus on parametric models, while in the next section we generalize to non-parametric models by parametric reduction. The main technical exercise involves the application of Theorems \ref{thm:main_thm} and \ref{thm:main_thm2} to statistical estimation with multiple quantized samples where the quantization of different samples can be independent or dependent as dictated by the communication protocol.

\subsection{Problem Formulation} Let
\begin{align*}
X_1, X_2, \cdots, X_n \overset{\text{i.i.d.}}{\sim} P_\theta, 
\end{align*}
where $\theta\in\Theta\subset \bR^d$. We consider three different types of protocols for communicating each of these samples with $k$ bits to a central processor that is interested in estimating the underlying parameter $\theta$:

\begin{itemize}
\item[(1)] Independent communication protocols $\Pi_{\mathsf{Ind}}$: each sample is independently quantized to $k$-bits and then communicated. Formally, for $i\in [n]$, each sample $X_i$ is encoded to a $k$-bit string $M_i$ by a possibly randomized quantization strategy, denoted by $q_{i}:\mathcal{X}\rightarrow[2^k]$, which can be expressed in terms of the conditional probabilities
$$
p(m_i|x_i)\quad \text{for }  m_i\in[2^k], \; x_i\in\mathcal{X}.
$$

\item[(2)] Sequential  communication protocols $\Pi_{\mathsf{Seq}}$: samples are communicated sequentially by broadcasting the communication to all nodes in the system including the central processor. Therefore, the quantization of the sample $X_i$ can depend on the previously transmitted quantized samples $M_1,\dots, M_{i-1}$ corresponding to  samples $X_1, \dots, X_{i-1}$ respectively. Formally, each sample $X_i$, for $i\in [n]$, is encoded to a $k$-bit string $M_i$ by  a set of possibly randomized quantization strategies $\{q_{m_1,\dots,m_{i-1}}:\mathcal{X}\rightarrow[2^k]: m_1,\dots, m_{i-1}\in[2^k]\}$, where each strategy $q_{m_1,\dots,m_{i-1}}(x_i)$ can be expressed in terms of the conditional probabilities
\begin{align*}
&p(m_i|x_i; m_1,\dots, m_{i-1})\\
& \quad \text{for } m_i\in[1:2^k]\,\text{and}\, x_i\in\mathcal{X}.
\end{align*}
\end{itemize}

While these two models can be motivated by a distributed estimation scenario where the topology of the underlying network can dictate the type of the protocol (see Figure~\ref{fig:sfig2}) to be used, they can also model the quantization and  storage of samples arriving sequentially at a single node. For example, consider a scenario where a continuous stream of samples is captured sequentially and each sample is stored in digital memory by using $k$ bits/sample. In the independent model, each sample would be quantized independently of the other samples (even though the quantization strategies for different samples can be different and jointly optimized ahead of time), while under the sequential model the quantization of each sample $X_i$ would depend on the information $M_1,\dots, M_{i-1}$ stored in the memory of the system at time $i$. This is illustrated in Figure~\ref{fig:sfig1}.

We finally introduce a third type of communication protocol that allows nodes to communicate their samples to the central processor in a fully interactive manner while still limiting the number of bits used per sample to $k$ bits. Under this model, each node can see the previously written bits on a public blackboard, and can use that information to determine its quantization strategy for subsequently transmitted bits. This is formally defined below. 

\begin{itemize}
\item[(3)] Blackboard communication protocols $\Pi_{\mathsf{BB}}$: all nodes communicate via a publicly shown blackboard while the total number of bits each node can write in the final transcript $Y$ is limited by $k$ bits. When one node writes a message (bit) on the blackboard, all other nodes can see the content of the message. Formally, a blackboard communication protocol $\Pi \in \Pi_{\mathsf{BB}}$ can be viewed as a binary tree \cite{kush}, where each internal node $v$ of the tree  is assigned a deterministic label $l_v\in [n]$ indicating the identity of the node to write the next bit on the blackboard if the protocol reaches tree node $v$; 
the left and right edges departing from $v$ correspond to the two possible values of this bit and are labeled by $0$ and $1$ respectively. Because all bits written on the blackboard up to the current time are observed by all nodes, the nodes can keep track of the progress of the protocol in the binary tree. The value of the bit written by node $l_v$ (when the protocol is at  node $v$ of the binary tree) can depend on the sample $X_{l_v}$ observed by this node (and implicitly on all bits previously written on the blackboard encoded in  the position of the node $v$ in the binary tree). Therefore, this bit can be represented by a function $b_v(x)=p_v(1|x)\in [0,1]$, which we associate with the tree node $v$; node $l_v$ transmits $1$ with probability $b_v(X_{l_v})$ and $0$  with probability $1-b_v(X_{l_v})$. Note that a proper labeling of the binary tree together with the collection of functions $\{b_v(\cdot)\}$ (where $v$ ranges over all internal tree nodes)  completely characterizes all possible (possibly probabilistic) communication strategies for the nodes.

The $k$-bit communication constraint for each node can be viewed as a labeling constraint for the binary tree; for each $i\in [n]$, each possible path from the root node to a leaf node can visit exactly $k$ internal nodes with label $i$. In particular, the depth of the binary tree is $nk$ and there is a one-to-one correspondence between all possible transcripts $y\in \{0,1\}^{nk}$ and paths in the tree. Note that there is also a one-to-one correspondence between $y\in \{0,1\}^{nk}$ and the $k$-bit messages $m_1,\dots, m_n$ transmitted by the $n$ nodes. In particular, the transcript $y\in \{0,1\}^{nk}$ contains the same amount of information as $m_1,\dots, m_n$, since given the transcript $y$ (and the protocol) one can infer $m_1,\dots, m_n$ and vice versa (for this direction note that the protocol specifies the node to transmit first, so given $m_1,\dots, m_n$ one can deduce the path followed in the protocol tree). 


\end{itemize}

\begin{figure}
\begin{subfigure}{.5\textwidth}
  \centering
  \includegraphics[width=.8\linewidth]{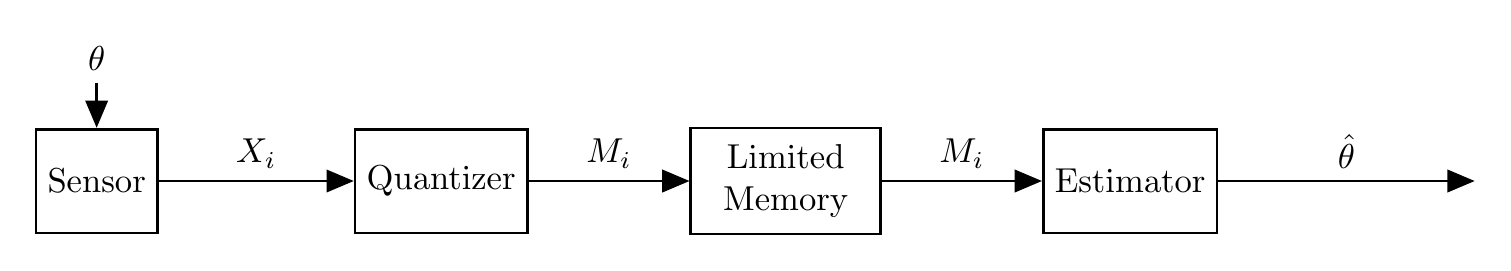}
  \caption{Storing a stream of samples}
  \label{fig:sfig1}
\end{subfigure}\\
\begin{subfigure}{.5\textwidth}
\centerline{\includegraphics[width=.8\linewidth]{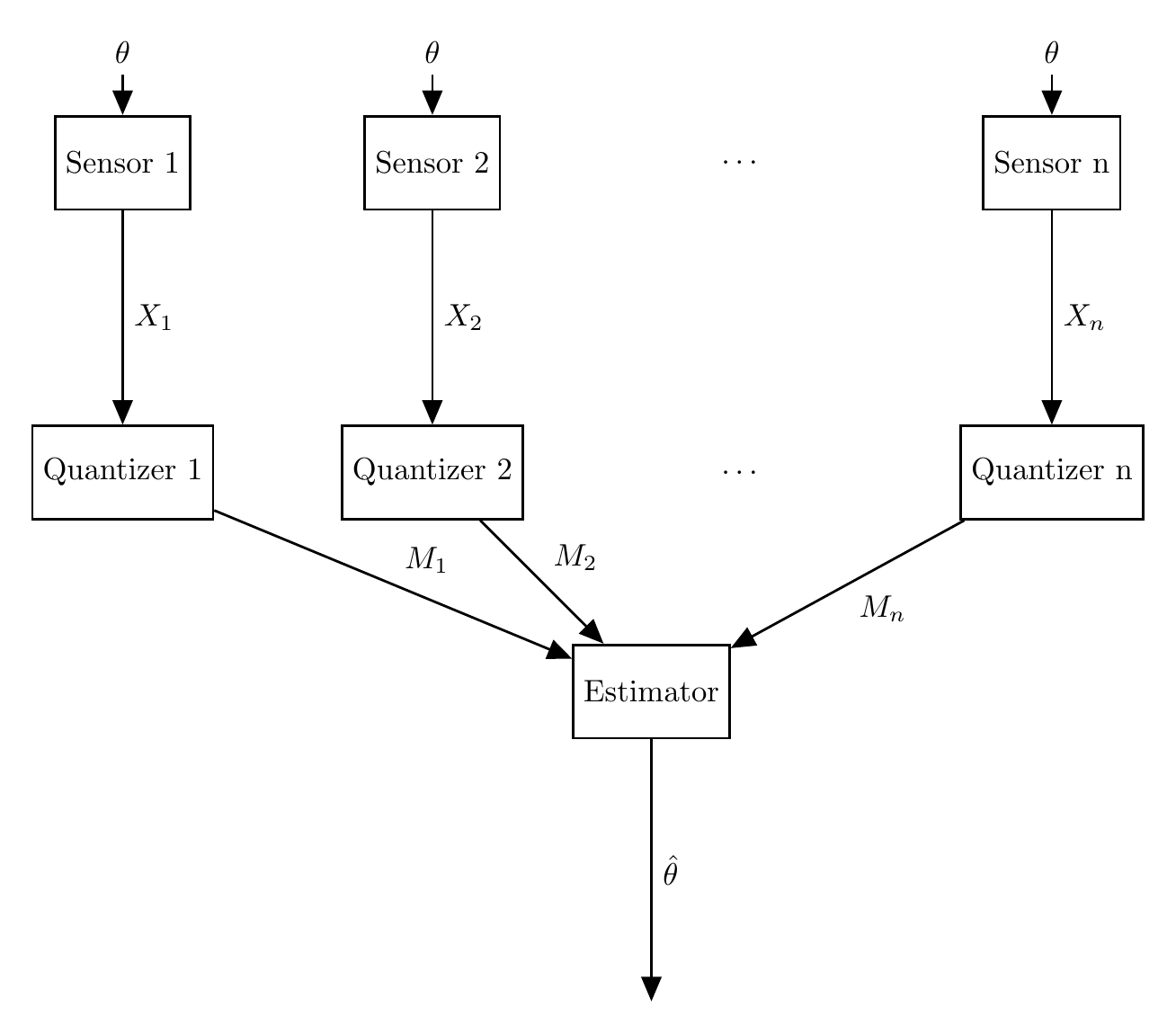}}
  \caption{Distributed communication of samples}
  \label{fig:sfig2}
\end{subfigure}
 \caption{Two applications that require quantization of samples. The quantization strategy can be independent or sequential.}
 \label{fig:fig}
\end{figure}



Under all three communication protocols above, the ultimate goal is to produce an estimate  $\hat{\theta}$ of the underlying parameter $\theta$ from the $nk$ bit transcript $Y$ or equivalently the collection of $k$-bit messages $M_1,\dots, M_n$ observed by the estimator. Note that the encoding strategies/protocols used in each case can be jointly optimized and agreed upon by all parties ahead of time. Formally, we are interested in the following parameter estimation problem under squared $\ell_2^2$ risk
$$
\inf_{(\Pi,\hat{\theta})}\sup_{\theta\in\Theta} \bE_{\theta}\|\hat{\theta}-\theta\|_2^2,
$$
where $\hat{\theta}(M_1,\dots, M_n)$ is an estimator of $\theta$ based on the quantized observations. Note that with an independent communication protocol, the messages $M_1,\ldots,M_n$ are independent, while this is no longer true under the sequential and blackboard protocols.

\subsection{Main Results for Distributed Parameter Estimation} 
We next state our main theorem for distributed parameter estimation. We will show in the next subsection that this theorem can be applied to obtain tight lower bounds for distributed estimation under many common statistical models, including the discrete distribution estimation and the Gaussian mean estimation.
\begin{thm} \label{thm:vant}
Suppose $[-B,B]^d\subset \Theta$. For any estimator $\hat{\theta}(M_1,\ldots,M_n)$ and communication protocol $\Pi \in\Pi_{\mathsf{Ind}}$, $\Pi_{\mathsf{Seq}}$, or $\Pi_{\mathsf{BB}}$, if $S_\theta(X)$ satisfies the hypotheses in Theorem \ref{thm:main_thm} then
\begin{align*}
\sup_{\theta\in\Theta} \mathbb{E} & \|\hat\theta - \theta\|_2^2 \nonumber \geq \frac{d^2}{I_02^kn + \frac{d\pi^2}{B^2}},
\end{align*}
and  if $S_\theta(X)$ satisfies the hypotheses in Theorem \ref{thm:main_thm2} then
\begin{align*}
\sup_{\theta\in\Theta} \mathbb{E} & \|\hat\theta - \theta\|_2^2 \nonumber \geq \frac{d^2}{CI_0k^{2/p}n + \frac{d\pi^2}{B^2}}, 
\end{align*}
where $C=4$.
\end{thm}
%

We next prove Theorem~\ref{thm:vant} for $\Pi \in\Pi_{\mathsf{Ind}}$ or $\Pi_{\mathsf{Seq}}$ by a straightforward application of the Van Trees Inequality combined with the conclusions of Theorems~\ref{thm:main_thm} and \ref{thm:main_thm2}. The proof for $\Pi \in \Pi_{\mathsf{BB}}$ requires more work and is deferred to the Appendix \ref{app:bbord}.
\bigbreak
\paragraph*{\textit{Proof of Theorem~\ref{thm:vant}}} We are interested in the quantity
$$
I_{(M_1,\ldots,M_n)}(\theta)
$$
under each model.
We have
\begin{align}
&\mathsf{Tr}(I_{(M_1,\ldots,M_n)}(\theta))=\sum_{j=1}^d  [I_{(M_1,\ldots,M_n)}(\theta)]_{j,j} \nonumber \\
& = \sum_{i=1}^n \sum_{j=1}^d [I_{M_i|(M_1,\ldots,M_{i-1})}(\theta)]_{j,j} \nonumber\\
& = \sum_{i=1}^n \sum_{m_1,\dots, m_{i-1}} p(m_1,\ldots,m_{i-1}|\theta) \mathsf{Tr}(I_{M_i|(m_1,\ldots,m_{i-1})}(\theta)) \label{eq:lower_bounds1}
\end{align}
due to the chain-rule for Fisher information. Under the independent model,
$$
[I_{M_i|(m_1,\ldots,m_{i-1})}(\theta)]_{j,j}=[I_{M_i}(\theta)]_{j,j}.
$$
Under the the sequential model, conditioning on specific $m_1,\ldots,m_{i-1}$ only effects the distribution $p(m_i|\theta)$ by fixing the quantization strategy for $X_i$. Formally, for the sequential model,
\begin{align*}
\mathbb{P}&(M_i=m_i|\theta; m_1,\ldots,m_{i-1})\\
&=\mathbb{P}(q_{m_1,\ldots,m_{i-1}}(X_i)=m_i|\theta; m_1,\ldots,m_{i-1})\\
&=\mathbb{P}(q_{m_1,\ldots,m_{i-1}}(X_i)=m_i|\theta),
\end{align*}
where the last step follows since $X_1,\dots, X_{i-1}$ is independent of $X_i$ and therefore conditioning of $m_1,\ldots,m_{i-1}$ does not change the distribution of $X_i$. Since the bounds from Theorems \ref{thm:main_thm} and \ref{thm:main_thm2} apply for any quantization strategy, they apply to each of the terms in \eqref{eq:lower_bounds1}, and the following statements hold under both quantization models:
\begin{itemize}
\item[(i)] Under the hypotheses in Theorem \ref{thm:main_thm},
$$\mathsf{Tr}(I_{M_1,\ldots,M_n}(\theta)) \leq nI_02^k \; .$$
\item[(ii)] Under the hypotheses in Theorem \ref{thm:main_thm2},
$$\mathsf{Tr}(I_{M_1,\ldots,M_n}(\theta)) \leq nCI_0k^\frac{2}{p} \; .$$
\end{itemize}

Consider the squared error risk in estimating $\theta$:
$$\mathbb{E}\|\theta-\hat\theta\|_2^2 = \sum_{i=1}^d \mathbb{E}[(\theta_i-\hat\theta_i)^2 ] \; .$$
In order to lower bound this risk, we will use the van Trees inequality \cite{gill}. Suppose we have a prior $\mu_i$ for the parameter $\theta_i$. For convenience denote $M=(M_1,\ldots,M_n)$. The van Trees inequality for the component $\theta_i$ gives
\begin{align}
\int_{-B}^B \mathbb{E} & [(\hat\theta_i(M) - \theta_i)^2] \mu_i(\theta_i)d\theta_i \nonumber \\
&\geq \frac{1}{\int_{-B}^B [I_M(\theta)]_{i,i}\mu_i(\theta_i)d\theta_i + I(\mu_i)}, \label{eq:van_trees}
\end{align}
where $I(\mu_i) = \int_{-B}^B \frac{\mu_i'(\theta)^2}{\mu_i(\theta)}d\theta$ is the Fisher information of the prior $\mu_i$. Note that the required regularity condition that $\mathbb{E}[S_{\theta_i}(M)] = 0$ follows trivially since the expectation over $M$ is just a finite sum:
$$\mathbb{E}[S_{\theta_i}(M)] = \sum_m \frac{\partial}{\partial\theta_i}p(m|\theta) = \frac{\partial}{\partial\theta_i}\sum_m p(m|\theta) = 0 \; .$$ 
The prior $\mu_i$ can be chosen to minimize this Fisher information and achieve $I(\mu_i) = \pi^2/B^2$ \cite{borovkov}. Let $\mu(\theta) = \prod_i \mu_i(\theta_i).$ By summing over each component,
\begin{align}
\int_\Theta \sum_{i=1}^d \mathbb{E} & [(\theta_i-\hat\theta_i)^2 ]\mu(\theta)d\theta \nonumber \\
 & \geq  \sum_{i=1}^d \frac{1}{\int_\Theta[I_{M}(\theta)]_{i,i}\mu(\theta)d\theta+ \frac{\pi^2}{B^2}} \label{eq:risk_bound1} \\
& \geq  d\ \frac{1}{\sum_{i=1}^d \frac{1}{d}\int_\Theta[I_{M}(\theta)]_{i,i}\mu(\theta)d\theta+ \frac{\pi^2}{B^2}} \label{eq:risk_bound2} \\
& =  \frac{d^2}{\int_\Theta \mathsf{Tr}(I_{M}(\theta))\mu(\theta)d\theta + \frac{d\pi^2}{B^2}}\nonumber \; .
\end{align}
Therefore,
\begin{align}
\sup_{\theta\in\Theta} \mathbb{E}  \|\hat\theta(M) - \theta\|^2 \geq \frac{d^2}{\sup_{\theta\in\Theta}\mathsf{Tr}(I_M(\theta)) + \frac{d\pi^2}{B^2}} \label{eq:main_ineq} \; .
\end{align}
The \eqref{eq:risk_bound1} follows from \eqref{eq:van_trees}, and  inequaltiy \eqref{eq:risk_bound2} follows from the convexity of $x\mapsto 1/x$ for $x>0$. We could have equivalently used the multivariate version of the van Trees inequality \cite{gill} to arrive at the same result, but we have used the single-variable version in each coordinate instead in order to simplify the required regularity conditions.

Combining \eqref{eq:main_ineq} with (i) and (ii) proves the theorem.

\subsection{Applications to Common Statistical Models}

Using the bounds in Section~\ref{sec:apps}, Theorem \ref{thm:vant} gives lower bounds on the minimax risk for the distributed estimation of $\theta$ under common statistical models. We summarize these results in the following corollaries. 
\begin{cor}[Gaussian location model]\label{cor.gaussian}
Let $X \sim\mathcal{N}(\theta,\sigma^2I_d)$ with $[-B,B]^d\subset \Theta$. For $nB^2\min\{k,d\} \geq d\sigma^2$, we have
\begin{align*}
\sup_{\theta\in\Theta} \mathbb{E}\|\hat{\theta}-\theta\|_2^2 \ge C\sigma^2 \max \left\{\frac{d^2}{nk}, \frac{d}{n}\right\}
\end{align*}
for any communication protocol $\Pi$ of type $\Pi_{\mathsf{Ind}}$, $\Pi_{\mathsf{Seq}}$, or $\Pi_{\mathsf{BB}}$ and any estimator $\hat{\theta}$, where $C>0$ is a universal constant independent of $n,k,d,\sigma^2,B$.
\end{cor}
\medbreak

Note that the condition  $nB^2\min\{k,d\} \geq d\sigma^2$ in the above corollary is a weak condition that ensures that we can ignore the second term in the denominator of \eqref{eq:main_ineq}. For fixed $B,\sigma,$ this condition is weaker than just assuming that $n$ is at least order $d$, which is required for a consistent estimation anyways. We will make similar assumptions in the subsequent corollaries.

The corollary recovers the results in \cite{duchi, garg} (without logarithmic factors in the risk)  and the corresponding result from \cite{yanjun} without the condition $k\geq\log d$. An estimator under the blackboard communication protocol which achieves this result is given in \cite{garg}.
\medbreak

\begin{cor}[Gaussian covariance estimation]\label{cor.variance}
Suppose that $X \sim\mathcal{N}(0,\text{\rm diag}(\theta_1,\ldots,\theta_d))$ with $[\sigma_{\min}^2,\sigma_{\max}^2]^d\subset\Theta$. Then for $n\left(\sigma_{\max}^2 - \sigma_{\min}^2\right)^2\min\{k^2,d\} \geq d\sigma_{\min}^4$, we have
\begin{align*}
\sup_{\theta\in\Theta} \mathbb{E}\|\hat{\theta}-\theta\|_2^2 \ge C\sigma_{\min}^4 \max \left\{\frac{d^2}{nk^2}, \frac{d}{n}\right\}
\end{align*}
for any communication protocol $\Pi$ of type $\Pi_{\mathsf{Ind}}$, $\Pi_{\mathsf{Seq}}$, or $\Pi_{\mathsf{BB}}$ and any estimator $\hat{\theta}$, where $C>0$ is a universal constant independent of $n,k,d,\sigma_\text{min},\sigma_\text{max}$.
\end{cor}
\medbreak
The bound in Corollary \ref{cor.variance} is new, and it is unknown whether or not it is order optimal.
\medbreak
\begin{cor}[Distribution estimation]\label{cor.dist}
Suppose that $\mathcal{X} = \{1,\ldots,d+1\}$ and that 
$$f(x|\theta) = \theta_x \; .$$
Let $\Theta$ be the probability simplex with $d+1$ variables. For $n\min\{2^k,d\} \geq d^2$, we have
\begin{align*}
\sup_{\theta\in\Theta} \mathbb{E}\|\hat{\theta}-\theta\|_2^2 \ge C\max \left\{\frac{d}{n2^k}, \frac{1}{n}\right\}
\end{align*}
for any communication protocol $\Pi$ of type $\Pi_{\mathsf{Ind}}$, $\Pi_{\mathsf{Seq}}$, or $\Pi_{\mathsf{BB}}$ and any estimator $\hat{\theta}$, where $C>0$ is a universal constant independent of $n,k,d$.
\end{cor}
\medbreak
This result recovers the corresponding result in \cite{yanjun} and matches the upper bound from the achievable scheme developed in \cite{yanjun2} (when the performance of the scheme is evaluated under $\ell_2^2$ loss rather than $\ell_1$).
\medbreak

\begin{cor}[Product Bernoulli model]\label{cor.bern}
Suppose that $X=(X_1,\ldots,X_d) \sim \prod_{i=1}^d\text{Bern}(\theta_i)$. If $\Theta = [0,1]^d,$ then for $n\min\{k,d\} \geq d$ we have
\begin{align*}
\sup_{\theta\in\Theta} \mathbb{E}\|\hat{\theta}-\theta\|_2^2 \ge C\max \left\{\frac{d^2}{nk}, \frac{d}{n}\right\}
\end{align*}
for any communication protocol $\Pi$ of type of type $\Pi_{\mathsf{Ind}}$, $\Pi_{\mathsf{Seq}}$, or $\Pi_{\mathsf{BB}}$ and any estimator $\hat\theta$, where $C>0$ is a universal constant independent of $n,k,d$.

If $\Theta = \{(\theta_1,\ldots,\theta_d)\in[0,1]^d : \sum_{i=1}^d \theta_i = 1\}$, then for $n\min\{2^k,d\} \geq d^2$, we get instead
\begin{align*}
\sup_{\theta\in\Theta} \mathbb{E}\|\hat{\theta}-\theta\|_2^2 \ge C\max \left\{\frac{d}{n2^k}, \frac{1}{n}\right\} \; .
\end{align*}
\end{cor}

The corollary recovers the corresponding result from \cite{yanjun} and matches the upper bound from the achievable scheme presented in the same paper. 

\section{Distributed Estimation of Non-Parametric Densities} \label{sec:nonpar}
Finally, we turn to the case where the $X_i$ are drawn independently from some probability distribution on $[0,1]$ with density $f$ with respect to the Lebesgue measure. We will assume that $f$ is H\"older continuous with smoothness parameter $s\in (0,1]$ and constant $L$, i.e.,
$$|f(x) - f(y)| \leq L|x - y|^s, \qquad \forall x,y\in [0,1].$$
Let $\mathcal{H}_L^s([0,1])$ be the space of all such densities. We are interested in characterizing the minimax risk
$$\inf_{(\Pi,\hat{f})}\sup_{f\in\mathcal{H}_L^s([0,1])} \mathbb{E}\|f-\hat{f}\|_2^2$$
where the estimators $\hat{f}$ are functions of the transcript $Y$. We have the following theorem.

\begin{thm}\label{thm:nonpar}
For any blackboard communication protocol $\Pi \in \Pi_{\mathsf{BB}}$ and estimator $\hat{f}(Y)$,
$$\sup_{f\in\mathcal{H}^s_L([0,1])} \mathbb{E}\|f-\hat{f}\|_2^2 \geq c\max\{n^{-\frac{2s}{2s+1}},(n2^k)^{-\frac{s}{s+1}}\} \; .$$
Moreover, this rate is achieved by an independent protocol $\Pi^\star\in \Pi_{\mathsf{Ind}}$ so that
$$\sup_{f\in\mathcal{H}_L^s([0,1])} \mathbb{E}\|f-\hat{f}\|_2^2 \leq C\max\{n^{-\frac{2s}{2s+1}},(n2^k)^{-\frac{s}{s+1}}\}, $$
where $c,C$ are constants that depend only on $s,L$.
\end{thm}

\begin{proof}
We start with the lower bound. Fix a bandwidth $h=1/d$ for some integer $d$, and consider a parametric subset of the densities in $\mathcal{H}^s_L([0,1])$ that are of the form
$$f_P(x) = 1 + \sum_{i=1}^d \frac{p_i-h}{h}g\left(\frac{x-x_i}{h}\right)$$
where $P=(p_1,\ldots,p_d)$, $x_i = (i-1)h$, and $g$ is a smooth bump function that vanishes outside $[0,1]$ and has $\int g(x)dx=1$. The function $f_P$ is in $\mathcal{H}^s_L([0,1])$ provided that $\max_{i\in [d]}|p_i - h| \leq c_0h^{s+1}$ for some constant $c_0$. Let $$\mathcal{P} = \left\{P \; : \; \sum_i p_i = 1, p_i\geq 0, |p_i-h|\leq c_0h^{s+1}\right\} \; .$$
For $P\in\mathcal{P}$ and any estimator $\hat{f}$ define $\hat{P} = (\hat{p}_1\ldots,\hat{p}_d)$ by $\hat{p}_i = \int_{x_i}^{x_{i+1}}\hat{f}(x)dx$. By Cauchy--Schwartz, we then have
\begin{align} \label{eq:nonparconverse}
\mathbb{E}\|f_P - \hat{f}\|_2^2 & = \mathbb{E}\left[\int_0^1(f_P(x)-f(x))^2dx\right] \nonumber \\
& = \mathbb{E}\left[\sum_{i=1}^d\int_{x_i}^{x_{i-1}}(f_P(x) - \hat{f}(x))^2dx\right]  \nonumber\\
& \geq d \;  \mathbb{E}\left[\sum_{i=1}^d\left(\int_{x_i}^{x_{i-1}}(f_P(x) - \hat{f}(x))dx\right)^2\right]  \nonumber\\
& = d \; \mathbb{E}\|P-\hat{P}\|_2^2 \; .
\end{align}
By the proofs of Theorem \ref{thm:main_thm} and Corollary \ref{cor.dist}, the quantity in \eqref{eq:nonparconverse} can be upper bounded provided that $h$ is not too small. In particular we can pick $h^{s+1} = (n\min\{2^k,d\})^{-\frac{1}{2}}$ so that
$$\sup_{P\in\mathcal{P}}\mathbb{E}\|f_P-\hat{f}\|_2^2 \geq  c\max\{n^{-\frac{2s}{2s+1}},(n2^k)^{-\frac{s}{s+1}}\}$$
as desired.

For the achievability side note that
$$\mathbb{E}\|f - \hat{f}\|_2^2 \leq 2\mathbb{E}\|f - f_h\|_2^2 + 2\mathbb{E}\|f_h - \hat{f}\|_2^2$$
and $f_h$ can be chosen to be a piece-wise constant function of the form
$$f_h(x) = \sum_{i=1}^d \frac{p_i}{h}1(x\in[x_i,x_{i+1})), $$
which satisfies
$$\|f - f_h\|_2^2 \leq C_0h^{2s}$$
for a constant $C_0$ that depends only on $L$. Choosing
$$\hat{f}(x) = \sum_{i=1}^d\frac{\hat{p}_i}{h}1(x\in[x_i,x_{i+1}))$$
for some $\hat{P} = (\hat{p}_1,\ldots,\hat{p}_d)$ we get
$$\mathbb{E}\|f - \hat{f}\|_2^2 \leq 2C_0h^{2s} + 2d\mathbb{E}\|P-\hat{P}\|^2 \; .$$
By using the estimator $\hat{P}$ (along with the specific communication protocol) from the achievability scheme in \cite{yanjun} for discrete distribution estimation, we are left with
$$\mathbb{E}\|f - \hat{f}\|_2^2 \leq 2C_0h^{2s} + C_1\frac{d^2}{n\min\{2^k,d\}} \; .$$
Optimizing over $h$ by setting $h^{2(s+1)} = (n\min\{2^k,d\})^{-1}$ gives the final result.
\end{proof}

\begin{appendix}

\subsection{Proofs of Theorems \ref{thm:main_thm} and \ref{thm:main_thm2}} \label{app:proofs}

Consider some $m$ and fix its likelihood $t=p(m|\theta)$. We will proceed by upper-bounding $\|\mathbb{E}[S_\theta(X)|m]\|_2$ from the right-hand side of \eqref{eq:lem2}. Note that
$$\mathbb{E}[S_\theta(X)|m] = \frac{\mathbb{E}[S_\theta(X)b_m(X)]}{t}$$
where $\mathbb{E}[b_m(X)] = t$ and $0\leq b_m(x)\leq 1$ for all $x\in\mathcal{X}$. We use $\langle\cdot,\cdot\rangle$ to denote the usual inner product. Let $U$ be a $d$-by-$d$ orthogonal matrix with columns $u_1,u_2,\ldots,u_d $ and whose first column is given by the unit vector
$$u_1 = \frac{1}{\| \mathbb{E}[S_\theta(X)|m] \|} \mathbb{E}[S_\theta(X)|m] \; .$$
We have
\begin{align*}
t\mathbb{E}[S_\theta(X)|m] & = \int S_\theta(x)b_m(x)f(x|\theta)d\nu(x) \nonumber \\
& = \int \left(\sum_{i=1}^d u_i\langle u_i, S_\theta(x)\rangle \right)b_m(x)f(x|\theta)d\nu(x) \nonumber\\
& = \sum_{i=1}^d \left(\int \langle u_i,S_\theta(x)\rangle b_m(x)f(x|\theta)d\nu(x)\right)u_i\label{eq:orthogexpansion1}
\end{align*}
and since $u_2,\ldots,u_d$ are all orthogonal to $\mathbb{E}[S_\theta(X)|m]$,
\begin{equation*} 
\mathbb{E}[S_\theta(X)|m] = \frac{1}{t}\left(\int \langle u_1,S_\theta(x)\rangle b_m(x)f(x|\theta)d\nu(x)\right)u_1 \; .
\end{equation*}
Therefore,
\begin{equation} \label{eq:generalcase1}
\| \mathbb{E}[S_\theta(X)|m]\|_2 = \frac{1}{t}\mathbb{E}[\langle u_1,S_\theta(X)\rangle b_m(X)] \; .
\end{equation}

\subsubsection{Proof of Theorem \ref{thm:main_thm}}
To finish the proof of Theorem \ref{thm:main_thm}, note that the upper bound $\mathsf{Tr}(I_M(\theta)) \leq \mathsf{Tr}(I_X(\theta))$ follows easily from the data processing inequality for Fisher information \cite{zamir}.
Using \eqref{eq:generalcase1} and the Cauchy-Schwarz inequality,
\begin{align*}
t\| \mathbb{E}[S_\theta(X)|m]\|_2^2 & = \frac{1}{t}\left(\mathbb{E}[\langle u_1,S_\theta(X)\rangle b_m(X)]\right)^2 \\
& \leq \frac{1}{t}\mathbb{E}[\langle u_1,S_\theta(X)\rangle^2]\mathbb{E}[b_m(X)^2] \\ 
& \leq \frac{1}{t}\mathbb{E}[\langle u_1,S_\theta(X)\rangle^2]\mathbb{E}[b_m(X)] \\ 
& = \mathbb{E}[\langle u_1,S_\theta(X)\rangle^2] \; .
\end{align*}
So if $\mathsf{Var}\langle u_1,S_\theta(X) \rangle \leq I_0$, then because score functions have zero mean,
$$t\| \mathbb{E}[S_\theta(X)|m]\|^2 \leq I_0 \; .$$
Therefore by Lemma 2,
$$\mathsf{Tr}(I_M(\theta)) \leq 2^kI_0 \; .$$

\subsubsection{Proof of Theorem \ref{thm:main_thm2}}

Turning to Theorem \ref{thm:main_thm2}, we now assume that for some $p\geq 1$ and any unit vector $u\in\mathbb{R}^d$, the random vector $\langle u,S_\theta(X)\rangle$ has finite squared $\Psi_p$ norm at most $I_0$. For $p=1$ or $p=2$, this is the common assumption that $S_\theta(X)$ is sub-exponential or sub-Gaussian, respectively, as a vector.

In particular $\|\langle u_1,S_\theta(X) \rangle\|_{\Psi_p}^2\leq I_0$, and the convexity of $x\mapsto \exp(|x|^p)$ gives
\begin{align*}
2 & \geq \mathbb{E}[\exp(|\langle u_1,S_\theta(X) \rangle|^p/I_0^{p/2})] \\
& \geq \mathbb{E}[b_m(X)\exp(|\langle u_1,S_\theta(X) \rangle|^p/I_0^{p/2})] \\
& = t\mathbb{E}[\exp(|\langle u_1,S_\theta(X) \rangle|^p/I_0^{p/2}|m] \\
& \geq t\exp\left( |\mathbb{E}[\langle u_1,S_\theta(X) \rangle | m]|^p / I_0^{p/2}\right)
\end{align*}
so that
$$\mathbb{E}[\langle u_1,S_\theta(X) \rangle|m] \leq \sqrt{I_0}\left(\log\left(\frac{2}{t}\right)\right)^\frac{1}{p} \; .$$
Therefore by \eqref{eq:generalcase1},
\begin{equation} \label{eq:tail_bound}
\|\mathbb{E}[S_\theta(X)|m]\|_2\leq \sqrt{I_0}\left(\log\left(\frac{2}{t}\right)\right)^\frac{1}{p} \; .
\end{equation}

By Lemma \ref{lem:lem2}, 
$$\mathsf{Tr}(I_M(\theta)) = \sum_m p(m|\theta)\|\mathbb{E}[S_\theta(X)|m]\|^2 \; ,$$
and therefore by \eqref{eq:tail_bound}, 
$$\mathsf{Tr}(I_M(\theta)) \leq I_0 \sum_m p(m|\theta) \left(\log\left(\frac{2}{p(m|\theta)}\right)\right)^\frac{2}{p} \; .$$
To bound this expression, let $\phi$ be the upper concave envelope of $x\mapsto x\left(\log\frac{2}{x}\right)^\frac{2}{p}$ on $[0,1]$. We have
\begin{align}
\mathsf{Tr}(I_M(\theta)) & \leq I_02^k\sum_m \frac{1}{2^k}\phi(p(m|\theta)) \nonumber \\
& \leq I_0 2^k \phi\left(\sum_m \frac{1}{2^k}p(m|\theta)\right) \nonumber \\
& = I_0 2^k \phi\left(\frac{1}{2^k}\right) \; .
\end{align}
It can be easily checked that $\phi(x) = x\left(\log\frac{2}{x}\right)^\frac{2}{p}$ for $0<x\leq1/2$, and therefore
\begin{align*}
\mathsf{Tr}(I_M(\theta)) & \leq I_0\left(\log 2^{k+1}\right)^\frac{2}{p} \\
& \leq I_0(k+1)^\frac{2}{p} \\
& \leq 4I_0k^\frac{2}{p} \; .
\end{align*}

\subsection{Proof of Corollaries \ref{cor:ex1} through \ref{cor:ex4}} \label{app:proofs2}

\subsubsection{Proof of Corollary \ref{cor:ex1}}
The score function for the Gaussian location model is
$$S_\theta(x) = \frac{1}{\sigma^2}(x-\theta)$$
so that $S_\theta(X)\sim \mathcal{N}\left(0,\frac{1}{\sigma^2}I_d\right)$. Therefore
$\langle u,S_\theta(X)\rangle$ is a mean-zero Gaussian with variance $1/\sigma^2$ for any unit vector $u\in\mathbb{R}^d$. Hence, the $\Psi_2$ norm of the projected score function vector is
$$\|\langle u,S_\theta(X)\rangle\|_{\Psi_2} = \frac{1}{\sigma}\sqrt{\frac{8}{3}}$$
since it can be checked that
$$\int \frac{\sigma}{\sqrt{2\pi}}e^{-\frac{\sigma^2x^2}{2}}e^{\left(\frac{x}{c}\right)^2}dx = 2$$
when $c = \frac{1}{\sigma}\sqrt{\frac{8}{3}} \; .$ By Theorem \ref{thm:main_thm2},
$$\mathsf{Tr}(I_M(\theta)) \leq \min\left\{\frac{d}{\sigma^2},C\frac{k}{\sigma^2}\right\}$$
where
$$C = \frac{32}{3} \; .$$

\subsubsection{Proof of Corollary \ref{cor:ex2}}
The components of the score function are
$$S_{\theta_i}(x) = \frac{x_i^2}{2\theta_i^2} - \frac{1}{2\theta_i} \; .$$
Therefore each independent component $S_{\theta_i}(X)$ is a shifted, scaled verson of a chi-squared distributed random variable $\chi_1^2$ with one degree of freedom. The $\Psi_1$ norm of each component $S_{\theta_i}(X)$ can be bounded by
\begin{equation} \label{eq:psi1ex}
\|S_{\theta_i}(X)  \|_{\Psi_1} \leq \frac{2}{\sigma^2_\text{min}} \; .
\end{equation}
This follows since the pdf of $Y = S_{\theta_i}(X)$ is
$$f_Y(y) = \frac{2\theta_i}{\sqrt{2\pi}}\frac{\exp\left(-(\theta_iy+\frac{1}{2})\right)}{\sqrt{2\theta_iy+1}} $$
for $y\geq -\frac{1}{2\theta_i}$, and we have
\begin{align*}
\int_{-\frac{1}{2\theta_i}}^\infty & \frac{2\theta_i}{\sqrt{2\pi}}\frac{\exp\left(-(\theta_iy+\frac{1}{2})\right)}{\sqrt{2\theta_iy+1}}\exp\left| \frac{y}{K}\right|dy \\
& = \sqrt{\frac{\theta_i}{\pi}} \int_{0}^\infty \frac{1}{\sqrt{y}}\exp\left(-\theta_iy + \left|\frac{y-\frac{1}{2\theta_i}}{K}\right|\right)dy \\
& \leq \sqrt{\frac{\theta_i}{\pi}} \int_{0}^\infty \frac{1}{\sqrt{y}}\exp\left(-\theta_iy + \left(\frac{y+\frac{1}{2\theta_i}}{K}\right)\right)dy \\
& = \sqrt{\frac{\theta_i}{\pi}} \exp\left(\frac{1}{2\theta_iK}\right) \int_{0}^\infty \frac{1}{\sqrt{y}}\exp\left(-\left(\theta_i-\frac{1}{K}\right) y\right)dy \; .
\end{align*}
By picking $K=\frac{2}{\theta_i}$ and using the identity $$\int_0^\infty \frac{1}{\sqrt{y}}e^{-cy}dy = \sqrt{\pi/c}$$ for $c>0$, we get
$$\mathbb{E}\left[\exp\left|\frac{Y}{K}\right|\right] \leq \sqrt{2}e^{\frac{1}{4}} \leq 2 \; .$$
This proves \eqref{eq:psi1ex}. We next turn our attention to bounding
$$\|\langle u,S_\theta(X)\rangle\|_{\Psi_1} $$
for any unit vector $u=(v_1,\ldots,v_d)\in\mathbb{R}^d$. To this end note that Markov's inequality implies
\begin{align}\label{eq:cor2markovsineq}
\mathbb{P}(|Y|\geq t) & = \mathbb{P}\left(e^\frac{|Y|}{K}\geq e^\frac{t}{K}\right) \leq 2e^{-\frac{t}{K}}, 
\end{align}
and we have the following bound on the moments of $Y$:
\begin{align} \label{eq:cor2momentbound}
\mathbb{E}[|Y|^p] & = \int_0^\infty \mathbb{P}(|Y|\geq t)pt^{p-1}dt \nonumber \\
& \leq \int_0^\infty 2e^{-\frac{t}{K}}pt^{p-1}dt \nonumber \\
& \leq 2p \int_0^\infty e^{-\frac{t}{K}}t^{p-1}dt = 2K^{p}\cdot p! \; .
\end{align}
This leads to a bound on the moment generating function for $Y$ as follows:
\begin{align} 
\mathbb{E}\left[e^{\frac{u_iY}{K}}\right] & \leq 1 + \mathbb{E}\left[\frac{u_iY}{K}\right] + \sum_{p=2}^\infty \mathbb{E}\left[\frac{1}{p!}\left(\frac{u_iY}{K}\right)^p\right] \nonumber \\ \label{eq:cor2mgfbound}
& \leq 1 + \sum_{p=2}^\infty 2|u_i|^p  = 1+\frac{2u_i^2}{1-|u_i|} \le e^\frac{2u_i^2}{1-|u_i|}
\end{align}
for $|u_i|<1$. Applying \eqref{eq:cor2mgfbound} to the moment generating function for $\langle u,S_\theta(X)\rangle$ gives
\begin{align}
\mathbb{E} & \left[\exp \left( \sum_{i=1}^d\frac{u_iS_{\theta_i}(X)}{K} \right)\right]  \nonumber \\
& = \exp\left(\frac{u_{i_0}S_{\theta_{i_0}}(X)}{K}\right)\prod_{i\neq i_0} \exp\left(\frac{u_iS_{\theta_i}(X)}{K}\right) \nonumber \\
& \leq 2e^{2(2+\sqrt{2})}, 
\end{align}
where $i_0$ is the only index such that $|u_{i_0}|> \frac{1}{\sqrt{2}}$ (if one exists). Thus
\begin{align} \label{eq:cor2absbound}
\mathbb{E} & \left[\exp{\left|\sum_{i=1}^d\frac{u_iS_{\theta_i}(X)}{K}\right|}\right]  \nonumber \\
& \leq \mathbb{E} \left[\exp \left( \sum_{i=1}^d\frac{u_iS_{\theta_i}(X)}{K}\right) \right] + \mathbb{E} \left[\exp\left(\sum_{i=1}^d\frac{-u_iS_{\theta_i}(X)}{K}\right)\right] \nonumber\\
& \leq 4e^{2(2+\sqrt{2})}, 
\end{align}
and by the concavity of $x\mapsto x^\alpha$ for $x\ge 0$ and $\alpha\in [0,1]$, 
\begin{align}
\mathbb{E} \left[\exp{\left|\sum_{i=1}^d\frac{u_iS_{\theta_i}(X)}{K/\alpha}\right|}\right] &\le \left(\mathbb{E} \left[\exp{\left|\sum_{i=1}^d\frac{u_iS_{\theta_i}(X)}{K}\right|}\right]\right)^{\alpha} \nonumber \\
& \leq (4e^{2(2+\sqrt{2})})^\alpha = 2, 
\end{align}
where $\alpha = \frac{\log 2}{\log 4 + 2(2+\sqrt{2})}$. Hence, we see that
$$\|\langle u,S_\theta(X)\rangle\|_{\Psi_1} \leq \frac{2}{\sigma^2_\text{min}\alpha}, $$
and using Theorem \ref{thm:main_thm2},
$$\mathsf{Tr}(I_M(\theta)) \leq \frac{16(\log 4 + 2(2+\sqrt{2}))^2}{(\log 2)^2}\left(\frac{k}{\sigma_\text{min}^2}\right)^2 \; .$$
For the other bound note that $\mathsf{Var}(S_{\theta_i}(X)) = \frac{1}{2\theta_i^2}$ so that
$$\mathsf{Tr}(I_X(\theta)) \leq \frac{d}{2\sigma_\text{min}^4} \; .$$

\subsubsection{Proof of Corollary \ref{cor:ex3}}
Note that
$$\theta_{d+1} = 1-\sum_{i=1}^{d} \theta_i, $$
and
$$S_{\theta_i}(x) = \begin{cases} \frac{1}{\theta_i}, & \; x=i \\ -\frac{1}{\theta_{d+1}}, & \; x=d+1 \\ 0, & \; \text{otherwise} \end{cases}$$
for $i=1,\ldots,d$. Recall that by assumption we are restricting our attention to $\frac{1}{4d}\leq\theta_i\leq\frac{1}{2d}$ for $i=1,\ldots,d$. Then for any unit vector $u=(v_1,\ldots,v_{d})$,
\begin{align*}
\mathsf{Var} & (\langle u,S_\theta(X)\rangle) \\
& = \sum_{x=1}^{d+1} \theta_x \left( \sum_{i=1}^{d}u_iS_{\theta_i}(x)\right)^2 \nonumber \\
& = \theta_{d+1}\frac{1}{\theta_{d+1}^2}\left(\sum_{i=1}^d u_i \right)^2 + \sum_{x=1}^{d} \theta_x \left( \sum_{i=1}^{d}u_iS_{\theta_i}(x)\right)^2 \nonumber \\
& \leq 2d + \sum_{x=1}^d \theta_x u_x^2\frac{1}{\theta_x^2} \leq 6d \; .
\end{align*}
Finally by Theorem \ref{thm:main_thm},
$$\mathsf{Tr}(I_M(\theta)) \leq  6\min\{d^2,d2^k\} \; .$$

\subsubsection{Proof of Corollary \ref{cor:ex4}}
With the product Bernoulli model
$$S_{\theta_i}(x) = \begin{cases} \frac{1}{\theta_i}, & \; x_i = 1 \\ -\frac{1}{1-\theta_i}, & \; x_i=0 \end{cases}, $$
so that in the case  $\Theta = [1/2-\varepsilon,1/2+\varepsilon]^d$,
$$\mathbb{E}\left[\exp\left[\left(\frac{S_{\theta_i}(X)}{K}\right)^2\right] \right] \leq \exp\left(\frac{1}{(\frac{1}{2}-\varepsilon)^2K^2}\right), $$
and
$$\|S_{\theta_i}(X)\|_{\Psi_2} \leq \frac{1}{(\frac{1}{2} - \varepsilon)\sqrt{\log 2}} \; .$$
By the rotation invariance of the $\Psi_2$ norm \cite{versh}, this implies
$$\|\langle u,S_{\theta_i}(X)\rangle\|_{\Psi_2} \leq \frac{C}{(\frac{1}{2} - \varepsilon)\sqrt{\log 2}}$$
for some absolute constant $C$. Thus by Theorem \ref{thm:main_thm2}, 
$$\mathsf{Tr}(I_M(\theta)) \leq 4 \left(\frac{C}{(\frac{1}{2} - \varepsilon)\sqrt{\log 2}}\right)^2k \; .$$

For the other bound,
$$\mathsf{Tr}(I_X(\theta)) \leq d\left(\frac{1}{\frac{1}{2}-\varepsilon}\right)^2 \; .$$
Now consider the sparse case $\Theta = [(\frac{1}{2}-\varepsilon)\frac{1}{d},(\frac{1}{2}+\varepsilon)\frac{1}{d}]^d$. We have
\begin{align*}
\mathsf{Var}\left(S_{\theta_i}(X)\right) & = \frac{1}{\theta_i^2}\theta_i + \left(\frac{-1}{1-\theta_i}\right)^2(1-\theta_i) \\
& = \frac{1}{\theta_i} + \frac{1}{1-\theta_i} \\
& \leq \frac{2d}{\frac{1}{2}-\varepsilon} \; ,
\end{align*}
and therefore by independence,
$$\mathsf{Var}\left(\langle u,S_{\theta}(X)\rangle \right) \leq \frac{2d}{\frac{1}{2}-\varepsilon} \; .$$
Thus by Theorem \ref{thm:main_thm}
$$\mathsf{Tr}(I_M(\theta)) \leq \frac{2d}{\frac{1}{2}-\varepsilon}\min\{d,2^k\} \; .$$

\subsection{Proof of Theorem \ref{thm:vant} with Blackboard Model}\label{app:bbord}

In order to lower bound the minimax risk under the blackboard model, we will proceed by writing down the Fisher information from the transcript $Y$ that is described in Section \ref{sec:est}. Let $b_{v,y}(x_{l_v}) = b_v(x_{l_v})$ if the path $\tau(y)$ takes the ``1'' branch after node $v$, and $b_{v,y}(x_{l_v}) = 1 - b_v(x_{l_v})$ otherwise.
The probability distribution of $Y$ can be written as
$$\mathbb{P}(Y=y) = \mathbb{E}\left[\prod_{v\in\tau(y)}b_{v,y}(X_{l_v})\right],$$
so that by the independence of the $X_i$,
\begin{align*}
\mathbb{P}(Y=y) & = \prod_{i=1}^n\mathbb{E}\left[\prod_{v\in\tau(y) \, : \, l_v=i}b_{v,y}(X_i)\right]\\
& = \prod_{i=1}^n\mathbb{E}\left[p_{i,y}(X_i)\right], 
\end{align*}
where $p_{i,y}(x_i) = \prod_{v\in\tau(y)\,:\,l_v=i}b_{v,y}(x_i).$ The score for component $\theta_i$ is therefore
$$\frac{\partial}{\partial\theta_i}\log \mathbb{P}(Y=y) = \sum_{j=1}^n \frac{\mathbb{E}[S_{\theta_i}(X_j)p_{j,y}(X_j)]}{\mathbb{E}[p_{j,y}(X_j)]} \; .$$
To get the above display, integration and differentiation have to be interchanged just like in Lemma \ref{lem:lem1}. The Fisher information from $Y$ for estimating the component $\theta_i$ is then
\begin{align*}
\mathbb{E} & \left[\left(\frac{\partial}{\partial\theta_i}\log \mathbb{P}(Y=y)\right)^2\right] \\
& = \sum_{j,k,y} \mathbb{P}(Y=y)\frac{\mathbb{E}[S_{\theta_i}(X_j)p_{j,y}(X_j)]\mathbb{E}[S_{\theta_i}(X_k)p_{k,y}(X_k)]}{\mathbb{E}[p_{j,y}(X_j)]\mathbb{E}[p_{k,y}(X_k)]} \; .
\end{align*}
Note that when $j\neq k$ the terms within this summation are zero:
\begin{align} \label{eq:treeprobdist}
\sum_y & \mathbb{P}(Y=y)\frac{\mathbb{E}[S_{\theta_i}(X_j)p_{j,y}(X_j)]\mathbb{E}[S_{\theta_i}(X_k)p_{k,y}(X_k)]}{\mathbb{E}[p_{j,y}(X_j)]\mathbb{E}[p_{k,y}(X_k)]} \nonumber \\
& = \mathbb{E}\left[S_{\theta_i}(X_j)S_{\theta_i}(X_k)\sum_y\prod_{l=1}^np_{l,y}(X_l)\right] \nonumber \\
& = \mathbb{E}\left[S_{\theta_i}(X_j)S_{\theta_i}(X_k)\right] = 0 \; .
\end{align}
The step in \eqref{eq:treeprobdist} follows since $\prod_{l=1}^np_{l,y}(x_l)$ describes the probability that $Y=y$ for fixed samples $x_1,\ldots,x_n$, and thus $\sum_y\prod_{l=1}^np_{l,y}(x_l) = 1$.

A related nontrivial identity, that
\begin{equation} \label{eq:treelem}
\sum_y\prod_{i\neq j} \mathbb{E}[p_{i,y}(X)] = 2^k
\end{equation}
for each $j\in [n]$, will be required later in the proof. For a tree-based proof of this fact see \cite{yanjun}.

Returning to the Fisher information from $Y$ we have that
\begin{align}\label{eq:traceBBfisherinfo}
\sum_{i=1}^d \mathbb{E} & \left[\left(\frac{\partial}{\partial\theta_i}\log\mathbb{P}(Y=y)\right)^2\right] \nonumber \\
& = \sum_y \mathbb{P}(Y=y)\sum_{j} \left(\frac{\mathbb{E}[S_{\theta_i}(X_j)p_{j,y}(X_j)]}{\mathbb{E}[p_{j,y}(X_j)]}\right)^2 \; .
\end{align}
Let $\mathbb{E}_{j,y}$ denote the expectation with respect to the new density $$\frac{p_{j,y}(x_j)f(x_j|\theta)}{\mathbb{E}[p_{j,y}(X_j)]} \; ,$$ then we can simplify \eqref{eq:traceBBfisherinfo} as
\begin{align}\label{eq:traceBBfisherinfo2}
\sum_{i=1}^d \mathbb{E} & \left[\left(\frac{\partial}{\partial\theta_i}\log\mathbb{P}(Y=y)\right)^2\right] \nonumber \\
& = \sum_y \mathbb{P}(Y=y)\sum_{j} \| \mathbb{E}_{j,y}[S_\theta(X_j)]\|_2^2 \; .
\end{align}

\begin{itemize}
\item[(i)] Suppose that $\|\langle u,S_\theta(X)\rangle\|_{\Psi_p}^2 \leq I_0$ holds for any unit vector $u\in\mathbb{R}^d$. Then letting 
$$u = \frac{\mathbb{E}_{j,y}[S_\theta(X)]}{\|\mathbb{E}_{j,y}[S_\theta(X)]\|_2}, $$
by the convexity of $x\mapsto \exp(|x|^p)$ we have
\begin{align*}
2 & \geq \mathbb{E}[\exp(|\langle u,S_\theta(X) \rangle|^p/I_0^{p/2})] \\
& \geq \mathbb{E}[p_{j,y}(X)\exp(|\langle u,S_\theta(X) \rangle|^p/I_0^{p/2})] \\
& = \mathbb{E}[p_{j,y}(X)]\mathbb{E}_{j,y}[\exp(|\langle u,S_\theta(X) \rangle|^p/I_0^{p/2})] \\
& \geq \mathbb{E}[p_{j,y}(X)]\exp\left( \big|\bE_{j,y}[\langle u,S_\theta(X) \rangle]\big|^p / I_0^{p/2} \right) \\
& \geq \mathbb{E}[p_{j,y}(X)]\exp\left( \bE_{j,y}[\langle u,S_\theta(X) \rangle]^p / I_0^{p/2} \right), 
\end{align*}
and thus
$$ \| \mathbb{E}_{j,y}[S_\theta(X_j)]\|_2 \leq I_0^{1/2}\left(\log\frac{2}{\mathbb{E}[p_{j,y}(X)]}\right)^\frac{1}{p} \; .$$
Continuing from \eqref{eq:traceBBfisherinfo2},
\begin{align}\label{eq:traceBBfisherinfo3}
& \sum_{i=1}^d \mathbb{E} \left[\left(\frac{\partial}{\partial\theta_i}\log\mathbb{P}(Y=y)\right)^2\right] \nonumber \\
& \leq I_0\sum_y \mathbb{P}(Y=y)\sum_{j}\left(\log\frac{2}{\mathbb{E}[p_{j,y}(X)]}\right)^\frac{2}{p} \nonumber \\
& = I_0\sum_{y,j} \left(\prod_{i\neq j} \mathbb{E}[p_{i,y}(X)] \right) \mathbb{E}[p_{j,y}(X)]\left(\log\frac{2}{\mathbb{E}[p_{j,y}(X)]}\right)^\frac{2}{p}
\end{align}
Finally, by upper bounding \eqref{eq:traceBBfisherinfo3} with the upper concave envelope $\phi$ of $x \mapsto x\left(\log\frac{2}{x}\right)^\frac{2}{p}$ on $[0,1]$, and then combining \eqref{eq:treelem} with Jensen's inequality:
\begin{align*}\label{eq:traceBBfisherinfo4}
\sum_{i=1}^d \mathbb{E} & \left[\left(\frac{\partial}{\partial\theta_i}\log\mathbb{P}(Y=y)\right)^2\right] \\
& \leq I_0\sum_{y,j}\left(\prod_{i\neq j} \mathbb{E}[p_{i,y}(X)] \right) \phi\left( \mathbb{E}[p_{j,y}(X)] \right)\\
& \leq I_0\sum_{j}2^k\phi\left( \sum_y \frac{1}{2^k}\mathbb{P}(Y=y) \right) \\
& = I_0n(k+1)^\frac{2}{p} \\
&\leq 4I_0nk^\frac{2}{p} \; .
\end{align*}
\item[(ii)] Now suppose we instead just have the finite variance condition that
$$\mathsf{Var}(\langle u,S_\theta(X)\rangle) \leq I_0$$
for any unit vector $u\in\mathbb{R}^d$. Picking again
$$u = \frac{\mathbb{E}_{j,y}[S_\theta(X)]}{\|\mathbb{E}_{j,y}[S_\theta(X)]\|_2} \; ,$$
the Cauchy-Schwarz inequality implies
\begin{align*}
& \mathbb{E}[p_{j,y}(X)]  \| \mathbb{E}_{j,y}[S_\theta(X)]\|_2^2 \\
& = \frac{1}{\mathbb{E}[p_{j,y}(X)]}\left(\mathbb{E}[\langle u,S_\theta(X)\rangle p_{j,y}(X)]\right)^2 \\
& \leq \frac{1}{\mathbb{E}[p_{j,y}(X)]}\mathbb{E}[\langle u,S_\theta(X)\rangle^2]\mathbb{E}[p_{j,y}(X)^2] \\ 
& \leq \frac{1}{\mathbb{E}[p_{j,y}(X)]}\mathbb{E}[\langle u,S_\theta(X)\rangle^2]\mathbb{E}[p_{j,y}(X)] \\ 
& = \mathbb{E}[\langle u,S_\theta(X)\rangle^2] \le I_0 \; .
\end{align*}
Together with \eqref{eq:traceBBfisherinfo2} this gives
\begin{align*}
& \sum_{i=1}^d \mathbb{E}  \left[\left(\frac{\partial}{\partial\theta_i}\log\mathbb{P}(Y=y)\right)^2\right] \nonumber \\
& = \sum_y \mathbb{P}(Y=y)\sum_{j} \| \mathbb{E}_{j,y}[S_\theta(X_j)]\|^2 \\
& \leq \sum_{j,y} I_0 \prod_{i\neq j} \mathbb{E}[p_{i,y}(X)] \\
& = I_0 n 2^k \; .
\end{align*}
The last equality follows from \eqref{eq:treelem}.
\end{itemize}

In both the sub-Gaussian (i) and finite variance (ii) cases, we apply the van Trees inequality just as in the proof for the independent or sequential models to arrive at the final result.

\end{appendix}

\section{Acknowledgement}
This work was supported in part by NSF award CCF-1704624 and by the Center for Science of Information (CSoI), an NSF Science and Technology Center, under grant agreement CCF-0939370.

\bibliographystyle{IEEEtran}
\bibliography{di.bib}

\end{document}